\documentclass{article}

\usepackage{arxiv}
\usepackage[utf8]{inputenc} 
\usepackage{fontenc}    
\usepackage{hyperref}       
\usepackage{url}            
\usepackage{booktabs}       
\usepackage{amsfonts}       
\usepackage{nicefrac}       
\usepackage{microtype}      

\usepackage{graphicx}
\usepackage{amssymb}
\usepackage{verbatim}
\usepackage{tikz}
\usetikzlibrary{arrows,automata}
\usepackage{algpseudocode}
\usepackage{algorithm}
\usepackage{amsmath}
\usepackage{mathtools}
\usepackage{todonotes}
\newcommand{\ta}{\ensuremath{\mathtt{a}}}

\newcommand{\GraphProblemFancyName}{EDT}
\newcommand{\NP}{\textsf{NP}}
\newcommand{\PSPACE}{\textsf{PSPACE}}
\newcommand{\LOGSPACE}{\textsf{L}}
\newcommand{\NLOGSPACE}{\textsf{NL}}
\newcommand{\PTIME}{\textsf{P}}
\DeclareMathOperator{\lang}{\mathcal{L}}
\newcommand{\decdhm}{\ensuremath{\operatorname{dec}_\text{dir\_hyp\_mul}}}
\usetikzlibrary{arrows.meta}
\tikzset{>={Latex[width=3mm,length=3mm]}}

\newcommand{\intreg}{\ensuremath{\mathit{int_{\mathrm{Reg}}}}}

\newcommand{\ov}[2]{\genfrac{}{}{0pt}{}{#1}{#2}}
\usepackage{amsthm}
\theoremstyle{definition}
\newtheorem{definition}{Definition}
\theoremstyle{plain}
\newtheorem{theorem}{Theorem}
\newtheorem{corollary}{Corollary}
\theoremstyle{remark}
\newtheorem{remark}{Remark}

\title{From Decidability to Undecidability by Considering Regular Sets of Instances\thanks{The author was partially supported by DFG (FE 560 / 9-1).}}
\author{
	Petra Wolf \\
	Abteilung Informatikwissenschaften -
	Fachbereich 4\\
	Universität Trier -
	Germany\\
	\texttt{wolfp@informatik.uni-trier.de}\\
	\url{https://www.wolfp.net/}
}
\begin{document}
	\maketitle
\begin{abstract}
	We are lifting classical problems from single instances to regular sets of instances. The task of finding a positive instance of the combinatorial problem $P$ in a potentially infinite given regular set is equivalent to the so called \intreg-problem of $P$, which asks for a given DFA $A$, whether the intersection of $P$ with $\lang(A)$ is non-empty.
	The \intreg-problem generalizes the idea of considering multiple instances at once and connects classical combinatorial problems with the field of automata theory.
%
%
	While the question of the decidability of the $\intreg$-problem has been answered positively for several \NP\ and even \PSPACE-complete problems, we are presenting some natural problems even from \LOGSPACE\ with an undecidable $\intreg$-problem. We also discuss alphabet sizes and different encoding-schemes elaborating the boundary between problem-variants with a decidable respectively undecidable $\intreg$-problem.  
\end{abstract}
%
%
%
\section{Introduction and Motivation}
In many fields multiple problem instances are considered all at once and they are accepted if there is at least one positive instance among them. The instances are described through a strongly compressed representation. For instance, in \emph{graph modification problems}\footnote{A Dagstuhl seminar on ``Graph Modification Problems'' was held in 2014 \cite{bodlaender2014graph}} a graph $G$ together with several graph editing operations is given and one asks whether $G$ can be transformed into a graph $G'$ with a certain property using up to $n$ editing operations \cite{bodlaender2014graph,gao2010survey,LiuWanGuo2014}. Here, the graph $G$ represents the finite set of graphs which can be generated by $G$ using up to $n$ editing operations. 
Another example are problems with \emph{uncertainty} in the instance \cite{DBLP:conf/icadlt/ChaariCAT14,Hla2012} where some parameters of the instance are unknown and therefore stand for a variety of values.
Finding a positive instance among plenty of candidates is also a task in synthesis problems \cite{desel1996synthesis}. In \cite{desel1996synthesis} the authors generate a finite set of candidate Petri nets among which they search for a solution. The synthesis of an object with a certain property can be seen as the search for an object with this property among several candidates. 

%
A natural generalization of the task of finding a positive instance in a finite set of instances is to search in an \emph{infinite} set of instances. 
A well studied class of potentially infinite languages are the regular languages which are also in a compressed way represented by finite automata or regular expressions.
We call $A=(Q, \Sigma, \delta, q_0, F)$ a deterministic finite automaton (DFA for short) if $Q$ is a finite set of states, $\Sigma$ a finite alphabet, $\delta \colon Q \times \Sigma \to Q$ a total transition function, $q_0 \in Q$ the start state and $F \subseteq Q$ the set of final states. We generalize $\delta$ to words in the usual way. We denote the language accepted by $A$ with $\lang(A) = \{w \in \Sigma^* \mid \delta(q_0, w) \in F\}$.
Asking whether the accepted language of a DFA $A$ contains a positive instance of a problem $P$ is equivalent to asking whether the intersection $P \cap \lang(A)$ is non-empty. This question was introduced in~\cite{guler2018deciding} as the \emph{$\intreg$}-problem of $P$ or $\intreg(P)$ for a fixed problem $P$.
\begin{definition}[$\intreg(P)$]
	\textit{Given:} DFA $A$.
	\textit{Question:} Is $\lang(A) \cap P \neq \emptyset$?
\end{definition}
%
In~\cite{DBLP:journals/iandc/AndersonLRSS09,DBLP:journals/eik/HorvathKK87,Ito1988}, $\intreg(L)$ was studied for languages $L$ with low computational complexity, but which describe structural word-properties that have high relevance for combinatorics on words and formal language theory (e.g., set of primitive words, palindromes, etc.). There, (efficient) decision procedures are obtained.

The $\intreg$-problem has  been studied independently under the name \emph{regular realizability problem} $RR(L)$, where the \emph{filter language} $L$ plays the role of problem $P$  above, i.\,e., $RR(L) = \intreg(L)$ 
(see~\cite{DBLP:journals/iandc/AndersonLRSS09,DBLP:journals/corr/Rubtsov15,abs-1105-5894,DBLP:conf/csr/TarasovV11,Vyalyi11,Vyalyi13uniJournal,VyalyiR15}),  motivated by  computational complexity questions. 
The aim was to present with the $RR$-problem `a specific class of algorithmic problems that
represents complexities of all known complexity classes [. . .] in a unified way'~\cite{Vyalyi13uniJournal}.
It turned out that $RR$-problems are universal in the sense that for any problem
$P$, there exists an $RR$-problem $RR(L)$ with the same complexity (note that $P$
and $L$ are different languages). In~\cite{VyalyiR15}
the authors focused on context-free filter languages and presented languages $L$ for which $RR(L)$ is either \PTIME-complete,
\NLOGSPACE-complete or has an intermediate complexity. In~\cite{DBLP:conf/csr/TarasovV11} the decidability of
the $RR$-problem with filter languages over permutations of binary words was studied.

In contrast, the line of research in~\cite{guler2018deciding,Wolf:Thesis:2018,DBLP:conf/dcfs/000219,DBLP:journals/corr/abs-2003-05826} aims to use the \intreg-problem as a tool to get insights into classes of hard problem as for instance the class of \NP\ and \PSPACE-complete problems. While the decidability of $\intreg(P)$ for hard problems $P$ such as \textsc{SAT}~\cite{guler2018deciding}, \textsc{ILP}~\cite{DBLP:conf/dcfs/000219}, \textsc{Vertex Cover}~\cite{DBLP:journals/corr/abs-2003-05826} and \textsc{TQBF}~\cite{guler2018deciding} is known, we present in this work problems, with a complexity ranging from contained in \LOGSPACE\ to \PSPACE-completeness, with an undecidable \intreg-problem. 
These results indicate that the decidability of the \intreg-problem of a language does not directly coincide with its computational complexity. 
This study rises the natural question what for instance \NP-complete problems with a decidable \intreg problem have in common that separates them from \NP-complete problems with an undecidable \intreg-problem. 
%
We also examine for some problems the size of the input alphabet and the encoding scheme resulting in different decidability results of the considered \intreg-problem. 

This paper is structured as follows. First, we discuss machine languages for several complexity classes. Then, we consider the problems of bounded and corridor tiling, followed by bounded PCP. We will show that all of these problems have an undecidable \intreg-problem.
%
%
Next, we investigate the \PSPACE-complete problem of {\sc Equivalence of Regular Expressions} and prove that the problem in a shuffled encoding has an undecidable $\intreg$ problem. As the proof only uses the concatenation operator of regular expressions, we get the undecidability of \intreg\ of the so called {\sc String Equivalence Modulo Padding} problem in a shuffled encoding, which lies in \LOGSPACE. For this problem we will discuss different alphabet sizes and encoding schemes and show that all other considered variants of this problem have a decidable \intreg-problem. Finally, we present a graph problem on directed multi-hyper-graphs with an undecidable \intreg-problem. This contrasts the results in~\cite{DBLP:journals/corr/abs-2003-05826} where classes of graph problems with a decidable \intreg-problem are identified. 

We expect the reader to be familiar with regular languages and their description through finite automata and regular expressions. The reader should also be familiar with the complexity classes \LOGSPACE, \NLOGSPACE, \NP, and \PSPACE. We refer to the textbooks \cite{garey1979computers} and \cite{hopcroft1969formal} for details. 
%
%
%
%
\section{Machine Languages}
For several complexity classes, we can define \emph{machine languages} which are complete for their complexity classes. We will show that the following machine languages have an undecidable $\intreg$-problem. The \intreg-problem of the machine language for NP was already discussed in \cite{guler2018deciding} and is listed here for the sake of completeness.
\begin{definition}[Machine Language for NL]
	\ \\
	\textit{Given:} Encoded nondeterministic Turing machine $\langle M \rangle$, input-word $x$, and a string $a^n$ with $n\in \mathbb{N}$.\\
	\textit{Question:} Does $M$ accept $x$ visiting only $\log(n)$\,different tape-positions?\\
	\textit{Encoding:} $L_{\text{NL}} = \{\langle M \rangle \$ x \$ a^n \mid M \text{ is an NTM accepting } x\text{ in }\log(n)\text{ space}\}$.
\end{definition}
The language $L_{\text{NL}}$ is complete for the class \NLOGSPACE. Every language in \NLOGSPACE\ can be accepted by a non-deterministic Turing machine which is space-bounded by a function $f \in \Theta(\log)$. Since $f$ is logspace-constructible, there exists a deterministic TM $M_f$ which computes $f(n)$ on the input $0^n$ in logarithmic space. 
Hence, every fixed problem in \NLOGSPACE\ can be reduced to $L_{\text{NL}}$ by hard-wiring the NTM $M$ which decides the problem and is space-bounded by $f$, followed by the input word $w$ and a unary string of size $2^{f(|w|)}$. Note that $f(|w|)$ is  logarithmically smaller than $|w|$ and hence can be stored using $\log(|w|)$ many cells. A logarithmically space-bounded TM can compute an output string which is exponentially in the size of its used memory. As can be easily verified $L_{\text{NL}} \in$  \NLOGSPACE.

The Machine Language for \NP, in short $L_{\text{NP}}$ is defined analogously demanding that $x$ is accepted in $n$ steps, while the Machine Language for \PSPACE, in short $L_{\text{PSPACE}}$ demands $x$ to be accepted in $n$ space. With similar arguments $L_{\text{NP}}$ is complete for \NP\ and $L_{\text{PSPACE}}$ is complete for \PSPACE.

%
\begin{theorem}
	$\intreg(L_{\text{NL}})$, $\intreg(L_{\text{NP}})$, and $\intreg(L_{\text{PSPACE}})$ are undecidable.
\end{theorem}
\begin{proof}
	We give a reduction from the undecidable non-emptiness-problem for recursively enumerable sets \cite{hopcroft1969formal} defined as $L_{\neq \emptyset} := \{\langle M \rangle \mid M$ is a nondeterministic TM with  $\lang(M) \neq \emptyset \}$.
	Let $\langle M \rangle$ be an arbitrary encoded Turing machine with the input alphabet $\Sigma$. We define the regular language 
	$f(\langle M \rangle) := R := \{\langle M \rangle \$ x \$ a^n \mid x \in \Sigma^*, n \geq 0 \}.$ 
	Then,
	$f(\langle M \rangle) \in \intreg(L_{\text{NL}}) \Leftrightarrow R \cap L_{\text{NL}} \neq \emptyset \Leftrightarrow \lang(M) \neq \emptyset \Leftrightarrow \langle M \rangle \in L_{\neq \emptyset}$.
	The same holds for $L_{\text{NP}}$ and $L_{\text{PSPACE}}$.
	Since the emptiness-problem for recursive enumerable sets is undecidable, the undecidability of the problems $\intreg(L_{\text{NL}})$, $\intreg(L_{\text{NP}})$, and $\intreg(L_{\text{PSPACE}})$ follows.
\end{proof}

\section{Bounded and Corridor Tiling}
The next problem we want to investigate is about the \emph{tiling of the plane}. 
For a given set of tile types and a fixed corner tile, the question is to fill a plane with the given tiles under some conditions. While the problem for an infinite plane is undecidable \cite{Kari08,StraubingTiling}, it becomes \NP-complete if we restrict the plane to an $n\times n$-square and preset the tiles on the edges of the square; it becomes \PSPACE-complete if we only restrict the width  with preset tiles and ask for a finite height, such that the plane can be tiled according to the preset tiles \cite{van1997convenience}.

First, we will give a formal definition of the problem {\sc Bounded Tiling}. Then, we will show that this problem has an undecidable $\intreg$-problem by reducing $L_{\neq \emptyset}$ to the problem $\intreg$({\sc Bounded Tiling}).

A \emph{tile} is a square unit where each of the edges is labeled with a color from a finite set $C$ of colors. The color assignment is described by \emph{tile types}. A tile type is a sequence $t = (w, n, e, s)$ with $w,n,e,s \in C$ of four symbols representing the coloring of the left, top, right, and bottom edge color. We denote with $t_w$, respectively $t_n$, $t_e$, $t_s$, the first, respectively second, third, and forth entry of the tuple $t$. Tiles can be regarded as instances of tile types. A tile must not be rotated or reflected. In the following problem, we give a finite set of tile types as input. From every tile type arbitrary many tiles can be placed. The tiles have to cover up a square grid region such that adjacent edges have to have the same color. The grid comes with an \emph{edge coloring} 
which contains for each border of the square grind a sequence of colors presetting the adjacent color of tiles resting on the edge.
A \emph{tiling} is a mapping from the square grid region to a set of tile types. With $\langle T \rangle$ we denote a proper encoding of the tile type set $T$ and with $[n]$ we denote the set $\{1, 2, \dots, n\}$.
We call $f \colon [n] \times [n] \to T$ a \emph{tiling function} if for all $i, j \in [n]$ it holds that $f(i,j)_e=f(i+1,j)_w$ for $i < n$ and $f(i,j)_n=f(i,j+1)_s$ for $j<n$ meaning that adjacent edges of the tiles have the same color. Here, the bottom left square of a grid region is indexed by $(1,1)$. 

\begin{definition}[{\sc Bounded Tiling}]
	\ \\
	\textit{Given:} Finite set $T$ of tile types with colors from a finite color set $C$ and an $n \times n$ square grid region $V$ with a given edge coloring. \\
	\textit{Question:} Is there a tiling function $f \colon [n] \times [n] \to T$ that tiles $V$ extending the edge coloring?\\
	\textit{Encoding:} 
	$\langle T \rangle$ followed by an edge coloring $\$l\$t\$r\$b$, $l=l_1\#l_2\#\dots\#l_{n}$, with
		$t=t_1\#t_2\#\dots\# t_n$,	 
		$r=r_1\#r_2\#\dots\#r_{n}$,
		$b=b_1\#b_2\#\dots \# b_n$
		with $ l_i, t_i, r_i, b_i \in C$.
\end{definition}

Howard Straubing gives in his article ``Tiling Problems'' \cite{StraubingTiling} a reduction from the complement of the halting problem to the problem of tiling an infinite plane. Therefore, he gives an algorithm ``that takes input $\langle M \rangle$ and produces the associated $\langle T, c\rangle$'' (where $c$ is the given corner tile in the unrestricted case of the problem). The tiles represent every possible transition of the Turing machine and are constructed in a way that correctly tiled rows correspond to configurations of the given Turing machine. The four colors of the tiles also ensure that two adjacent rows represent two consecutive configurations. Therefore, the infinite plane can only be tiled if and only if the Turing machine runs forever.

Peter van Emde Boas \cite{van1997convenience} uses a similar construction to simulate Turing machines and shows that the {\sc Bounded Tiling} problem is \NP-complete.
For a given nondeterministic Turing machine, the possible transitions and tape cell labelings are transformed into a set of tile types.
The input word, padded with blank symbols, is encoded in the bottom edge coloring $b$ and a distinguished accepting configuration is encoded in the top edge coloring $t$. 
The left and right borders are colored with the fixed color \emph{white} which is a color only occurring on vertical edges and which do not represent any state or alphabet letter of the Turing machine. So, white can be seen as a neutral border color. 
Blank symbols are trailed to the input word to enlarge the size of the square field to the exact time bound of the Turing machine. The Turing machine is altered in a way that it accepts with one distinguished accepting configuration. The tile types are constructed in a way that this accepting configuration can be repeated over several adjacent rows. 
Therefore, the constructed edge colored square region can be correctly tiled matching the edge coloring if and only if the given Turing machine accepts the input word within its time bound.


With that construction in mind, we will now prove that the $\intreg$-problem for {\sc Bounded Tiling} is undecidable.

\begin{theorem}
	\label{thm:tiling}
	The problem $\intreg(${\sc Bounded Tiling}$)$ is undecidable.
\end{theorem}
\begin{proof}
	We give a reduction from the undecidable problem $L_{\neq \emptyset}$.
	Let $\langle M \rangle$ be the encoding of an arbitrary NTM. 
	We construct a regular language $R$ which contains a positive {\sc Bounded Tiling} instance if and only if $M$ accepts at least one word. 
	We alter the machine $M$ to an NTM $N$ which behaves like $M$ except having only one distinguished accepting configuration, i.e., an empty tape with the head on the first position of the former input word.
	According to Straubing \cite{StraubingTiling} and van Emde Boas \cite{van1997convenience}, there is an algorithm which, given a TM $N$, produces the corresponding set of tile types $T$ such that a correct extending tiling of a given edge colored square field corresponds to a sequence of successive configurations of the given machine, starting on an input word represented through the coloring of the bottom border.
	
	Let $T_N$ be the corresponding tile type set for the NTM $N$ and let $C_N$ be the set of colors appearing in $T_N$. 		
	Let $C_{\Sigma} \subseteq C_N$ be the subset of colors representing input alphabet symbols, let $\diamond \in C_N$ be the \emph{white} color representing a white vertical border edge of the square grid, and let $\square \in C_N$ be the color representing an empty tape cell. Finally let $q_f \in C_N$ be the color representing the accepting state of the Turing machine.
	We define the regular set $R$ as
	$R = \lang\left(\{\langle T_N \rangle\ \$\  
	{\diamond}^*\ \$\ 
	q_f{\square}^*\ \$\ 
	{\diamond}^*\ \$\
	{C_{\Sigma}}^*{\square}^*\}\right).$
	The set $R$ consists of the set of tile types for the NTM $N$ together with edge colorings for every possible input word and every possible size of the field $V$. The top row will always contain the accepting configuration of $N$ padded with arbitrary many blank symbols. 
	The left and right borders of the field $V$ can consist of arbitrary many white edges, while the bottom row can encode every possible input word with arbitrary many added blank symbols allowing an arbitrary time bound for the Turing machine.
	Note that the edge coloring does not have to define a square, but the square shape is also contained in the set $R$ for every input word and every number of padding symbols. 
	Therefore, for every input word $w$, the set $R$ contains every size of squared fields with $w$ encoded in the bottom edge coloring. The tile type set of $R$ is constructed in a way that in a valid tiling adjacent rows will represent successive configurations of the Turing machine. So, for every number of steps the TM makes on the input word, there is a square field, with the input word encoded, in the set $R$ preventing enough space for the configurations of the TM. This brings us to our main claim,
	$R \cap \text{\sc Bounded Tiling} \neq \emptyset \Leftrightarrow \lang(N) \neq \emptyset$.
\end{proof}

With the same argument, we can show that the \PSPACE-complete problem {\sc Corridor Tiling} \cite{van1997convenience} also has an undecidable $\intreg$-problem.
\begin{definition}[{\sc Corridor Tiling}]
	\ \\
	\textit{Given:} Finite set $T$ of tile types with colors from a finite color set $C$, a top edge coloring $t$, and a bottom edge coloring $b$ of a grid region $V$, both of length $n$.\\
	\textit{Question:} Is there a finite height $h$ and a tiling function $f \colon [n] \times [h] \to T$ that tiles $V$\,extending the edge coloring.
\end{definition}
\begin{corollary}
	\label{cor:cor_tiling}
	The problem $\intreg(\text{\sc Corridor Tiling})$ is undecidable.
\end{corollary}
\begin{proof}
	The proof works analogously to the proof of Theorem \ref{thm:tiling} with a reduction from $L_{\neq \emptyset}$, the only difference is that $R$ only encodes the bottom and top borders and no left and right borders. 
\end{proof}

\section{Bounded PCP}
Another undecidable problem, which becomes decidable if we restrict the size of the potential solution, is the {\sc Post's Correspondence Problem} (in short {\sc PCP}). We show that the \NP-complete version {\sc Bounded Post Correspondence Problem} \cite{garey1979computers} (in short {\sc BPCP}) has an undecidable $\intreg$-problem by a reduction from the unrestricted undecidable {\sc PCP} problem \cite{hopcroft1969formal}. 
\begin{definition}[{\sc BPCP}]
	\ \\
	\textit{Given:} Finite alphabet $\Sigma$, two sequences $A=(a_1, a_2, \dots, a_n)$, $B=(b_1, b_2, \dots, b_n)$ of strings from $\Sigma^*$, and a positive integer $K \leq n$.\\
	\textit{Question:} Is there a sequence $i_1, i_2, \dots, i_k$  of $k \leq K$ (not necessarily distinct) positive integers $i_j \in [n]$ such that $a_{i_1}a_{i_2}\dots a_{i_k} = b_{i_1}b_{i_2}\dots b_{i_k}$?\\
	\textit{Encoding:} $L_{BPCP} := \{a_1\#a_2\#\dots\#a_n\$b_1\#b_2\#\dots\#b_n\$\operatorname{bin}(K) \mid K \leq n \wedge A=(a_1, a_2, \dots, a_n), B=(b_1, b_2, \dots, b_n) \text{ is a PCP instance with a solution} \leq K \}$.
\end{definition}
\noindent The problem {\sc PCP} is defined analogously but does not contain a bound $K$ on the size of a solution.
\begin{theorem}
	The problem $\intreg$$(${\sc BPCP}$)$ is undecidable.
\end{theorem}
\begin{proof}
	We give a reduction {\sc PCP} $\leq \intreg$$(${\sc BPCP}$)$.
	Let $A=(a_1, a_2, \dots, a_n)$ and $B=(b_1, b_2, \dots, b_n)$ be a {\sc PCP} instance. We construct a regular language~$R$ consisting of the given {\sc PCP} instance combined with every possible solution bound~$K$. Since $K$ is bounded by the length of list $A$ and $B$, we will pump those lists up by repeating the last list element of both lists arbitrarily often. Because the same element can be picked multiple times, adding elements already appearing in the given lists does not change the solvability of the instance. We define $R$ as
	$R= \{a_1\#a_2\#\dots\#a_n(\#a_n)^*\$b_1\#b_2\#\dots\#b_n(\#b_n)^*\$\{0,1\}^*\}.$
	It holds that $R\,\cap$\,{\sc BPCP} $\neq \emptyset$ if and only if there is a sequence of indexes $i_1, i_2, \dots, i_m$ such that $a_{i_1}a_{i_2}\dots a_{i_m} = b_{i_1}b_{i_2}\dots b_{i_m}$.
%
\end{proof}


\section{Regular Expressions in a Shuffled Encoding}
In this next section we show that the problem of {\sc Equivalence of Regular Expressions} (in short $\equiv_{\text{\sc RegEx}}$) over a binary alphabet in a shuffled encoding has an undecidable regular intersection emptiness problem. It turns out, that the problem is already undecidable if the regular expressions do not use alternation or the Kleene star. Thus, also the problem of {\sc String Equivalence Modulo Padding} over a binary alphabet in a shuffled encoding has an undecidable $\intreg$-problem. 
When we consider the {\sc String Equivalence Modulo Padding} problem over a unary alphabet or in a sequential encoding, the problem becomes decidable.
We first define the problem of {\sc Equivalence of Regular Expressions} (adapted from \cite{garey1979computers}).
For a regular expression $E$ we denote with $\lang(E)$ the regular language described by $E$. We use concatenation implicitly and omit the operator symbol. The alternation is represented by $|$-symbols.
\begin{definition}[{\sc Shuffled$\equiv_{\text{RegEx}}$}]
	\ \\
	\textit{Given:} A word $w = e_{1}f_{1}e_{2}f_{2}e_{3}f_{3}\dots e_{n}f_{n}$ over the alphabet $\Sigma \cup \{\emptyset, \epsilon, (,),|,*\}$ such that $E = e_{1}e_{2}e_{3}\dots e_{n}$ and $F = f_{1}f_{2}f_{3}\dots f_{n}$ are regular expressions over the alphabet $\Sigma$ using the operators \emph{alternation}, \emph{concatenation}, and \emph{Kleene star}. Note that one regular expression can be padded with $\epsilon$ or $\emptyset$ if the regular expression are of unequal length. \\
	\textit{Question:} Is $\lang(E) = \lang(F)$?
\end{definition}
The problem of equivalence of the regular expressions is well known to be \PSPACE-complete \cite{stockmeyer1973word}. Since we can change the encoding of an $\equiv_{\text{\sc RegEx}}$
 instance from shuffled to sequential and vice versa in quadratic time, the shuffled version of this problem is also \PSPACE-complete.
We will show that $\intreg$({\sc Shuffled$\equiv_{\text{RegEx}}$}) is undecidable by a reduction from the {\sc PCP} problem~\cite{hopcroft1979introduction}.
For readability reasons, we will refer to words $w \in$ {\sc Shuffled$\equiv_{\text{RegEx}}$} as $w = \ov{e_{1}}{f_{1}} \ \dots\ \ov{e_{n}}{f_{n}}$.
%
%
From a given {\sc PCP} instance we will construct a regular language $L_\text{Reg}$, the words of which will describe all possible solutions of the {\sc PCP} instance. 
The words will consist of two shuffled regular expressions using only the concatenation as an operator. 
By construction, the first regular expression will be a concatenation of strings from the $A$ list of the {\sc PCP} instance while the second regular expression will consists of the concatenated corresponding strings from the $B$ list. 
Since the regular expressions only use concatenation, languages described by them only contain one element each.
The language $L_\text{Reg}$ will contain two shuffled regular expressions describing the same language if and only if the {\sc PCP} instance has a valid solution.
\begin{theorem}\label{thm:regex}
	The problem $\intreg(${\sc Shuffled$\equiv_{\text{RegEx}}$}$)$ is undecidable. 
\end{theorem}
\begin{proof}
	We give a reduction {\sc PCP} $\leq \intreg(\text{\sc Shuffled}\equiv_{\text{RegEx}})$ and translate a given {\sc PCP} instance into a regular language $L_\text{Reg}$. We emphasize references to the \emph{regular expression} defining the language $L_\text{Reg}$, while references to the regular expressions encoded in the words of $L_\text{Reg}$ are not emphasized. We also emphasize references to the \emph{regular language} of shuffled regular expressions.
	
	Let $A = a_1, a_2, \dots, a_k$ and $B = b_1, b_2, \dots, b_k$ be a {\sc PCP} instance. 
	We define a \emph{regular expression}, describing a \emph{regular language} $L_\text{Reg}$ of shuffled regular expressions describing concatenations of list elements.
	Let $L_\text{Reg}$ be defined through the \emph{regular expression}
	$$\left(\left. \ov{{a_1}'}{{b_1}'} \right|\left. \left.\ov{{a_2}'}{{b_2}'} \right| \dots \right| \ov{{a_k}'}{{b_k}'} \right)^+$$
	where
	the string $\ov{{a_i}'}{{b_i}'}$ consists of the two shuffled strings $a_i$, $b_i$ where the shorter string is padded with $\epsilon$-symbols at the end until both stings have the same length.
	The $\epsilon$-symbol is here used as an alphabet symbol of the language $L_\text{Reg}$ and refers to the regular expression $\epsilon$ which will be interpreted as $\{\epsilon\}$ and not to the empty word itself.
	Therefore, $L_\text{Reg}$ consists of all possible pairwise concatenations of elements of the lists $A$ and $B$ where the concatenated strings are padded with $\epsilon$-symbols to have the same length.
	
	For every {\sc PCP} instance the described \emph{regular expression} of the language $L_\text{Reg}$ can be computed by a computable total function. It remains to show that the {\sc PCP} instance $A$, $B$ has a solution if and only if $L_\text{Reg}\, \cap\, ${\sc Shuffled$\equiv_{\text{RegEx}}$} $ \neq \emptyset$. More precisely, the intersection will contain all solutions of the {\sc PCP} instance.
	
	First, consider the only if direction. Let $i_1, i_2, \dots , i_n$ be a solution of the {\sc PCP} instance $A, B$ such that $a_{i_1}a_{i_2}\dots a_{i_n} = b_{i_1}b_{i_2}\dots b_{i_n}$. By construction, the \emph{regular language} $L_\text{Reg}$ contains all possible concatenations of the strings $\ov{{a_1}'}{{b_1}'}, \dots, \ov{{a_k}'}{{b_k}'}$ corresponding to the pairs $(a_1, b_1), \dots (a_k, b_k)$ of the strings of the lists $A$ and $B$. Therefore, $L_\text{Reg}$ also contains the word 
	$w = \ov{{a_{i_1}}'}{{b_{i_1}}'}\ov{{a_{i_2}}'}{{b_{i_2}}'}\dots \ov{{a_{i_n}}'}{{b_{i_n}}'}$.
	The word $w$ consists of the two shuffled regular expressions $E = {a_{i_1}}'{a_{i_2}}'\dots {a_{i_k}}'$ and $F = {b_{i_1}}'{b_{i_2}}'\dots {b_{i_k}}'$. Since they are both nonempty strings with padded $\epsilon$'s their described language is a singleton set. By construction, we have $\lang(E) = \{{a_{i_1}}{a_{i_2}}\dots {a_{i_k}}\}$ and $\lang(F) = \{{b_{i_1}}{b_{i_2}}\dots {b_{i_k}}\}$. By assumption is $a_{i_1}a_{i_2}\dots a_{i_n} = b_{i_1}b_{i_2}\dots b_{i_n}$, therefore we have $\lang(E) = \lang(F)$ and $w \in L_{Reg} \cap $ {\sc Shuffled$\equiv_{\text{RegEx}}$}.
	
	For the other direction, assume $L_\text{Reg} \cap $ {\sc Shuffled$\equiv_{\text{RegEx}}$} $\neq \emptyset$. Let $w = \ov{{a_{i_1}}'}{{b_{i_1}}'}\ov{{a_{i_2}}'}{{b_{i_2}}'}\dots \ov{{a_{i_n}}'}{{b_{i_n}}'} \in  L_\text{Reg} \cap $ {\sc Shuffled$\equiv_{\text{RegEx}}$} 
	consists of the two shuffled regular expressions $E = {a_{i_1}}'{a_{i_2}}'\dots {a_{i_k}}'$ and $F = {b_{i_1}}'{b_{i_2}}'\dots {b_{i_k}}'$. By assumption is $\lang(E) = \lang(F)$. Since $\lang(E)$ and $\lang(F)$ each contain only one element, from which the describing regular expressions differ only by padded $\epsilon$-symbols, it holds by construction that $a_{i_1}a_{i_2}\dots a_{i_n} = b_{i_1}b_{i_2}\dots b_{i_n}$. Therefore, $i_1, i_2, \dots , i_n$ is a solution of the {\sc PCP} instance.
\end{proof}

To show undecidability of the $\intreg$({\sc Shuffled$\equiv_{\text{RegEx}}$}) problem we have  made use of only one operator of regular expressions, namely the concatenation. If we restrict the {\sc Shuffled$\equiv_{\text{RegEx}}$} problem to regular expressions using only letters from $\Sigma$, the $\epsilon$-symbol and the concatenation, we get the much easier problem of {\sc Shuffled String Equivalence Modulo Padding}, in short {\sc Shuffled$\equiv_{\text{String}^\epsilon}$}.
Since we are only using the associative operation of concatenation, we can get rid of brackets.
All of the following problems are in the complexity class \LOGSPACE, since they can be solved deterministically using two pointers.

\begin{definition}[{\sc Shuffled$\equiv_{\text{String}^\epsilon}$}]
	\ \\
	\textit{Given:} A word $w = s_{1}t_{1}s_{2}t_{2}s_{3}t_{3}\dots s_{n}t_{n}$ such that $s_i, t_i \in \Sigma \cup \{\epsilon \}$.
	\textit{Question:} Is $h(s_{1}s_{2}s_{3}\dots s_{n}) = h(t_{1}t_{2}t_{3}\dots t_{n})$ where $h: \left(\Sigma \cup \{\epsilon \}\right)^* \to \Sigma^*$ is an erasing homomorphism which leaves all symbols in $\Sigma$ unchanged and deletes the $\epsilon$-symbols.
\end{definition}

\begin{corollary}\label{thm:string}
	The problem \intreg$(${\sc Shuffled$\equiv_{\text{String}^\epsilon}$}$)$ is undecidable.
\end{corollary}
\begin{proof}
	In the proof of Theorem \ref{thm:regex} we have constructed a \emph{regular language} of shuffled regular expressions which described singleton sets by using concatenation and padding with $\epsilon$-symbols. So, the two regular expressions describe the same language if and only if the regular expressions themselves yield the same string under deleting the $\epsilon$-symbols. Therefore, the proof of Theorem \ref{thm:regex} also works for Theorem \ref{thm:string}.
\end{proof}


If we restrict the alphabet $\Sigma$ to singleton sets, the \textsc{Shuffled$\equiv_{\text{String}^\epsilon}$} becomes decidable as this problem, considered as a language, is a context-free language. Alternatively, if we refrain from the shuffled encoding and consider instead a sequential encoding, the problem also becomes decidable. Here, we identify sub-automata which accept prefixes up to the symbol \$ and sub-automata which accept suffixes starting after the symbol \$. We use a homomorphism $h$ to erase the padding symbol $\epsilon$ and check for each pair of prefix and suffix sub-automata $A_P$ and $A_S$ whether $h(\lang(A_P)) \cap h(\lang(A_S)) \neq \emptyset$. 
\begin{definition}[{\sc Unary-Shuffled$\equiv_{\text{String}^\epsilon}$}]
	\ \\
	\textit{Given:} A word $w = s_{1}t_{1}s_{2}t_{2}s_{3}t_{3}\dots s_{n}t_{n}$ such that $s_i,t_i \in \Sigma \cup \{\epsilon \}$ with $\lvert \Sigma \rvert = 1$.\\
	\textit{Question:} Is $h(s_{1}s_{1}s_{1}\dots s_{1}) = h(t_{1}t_{2}t_{3}\dots t_{n})$ where $h: \left(\Sigma \cup \{\epsilon \}\right)^* \to \Sigma^*$ is an erasing homomorphism which leaves all symbols in $\Sigma$ unchanged and deletes the $\epsilon$-symbols.
\end{definition}
\begin{theorem}\label{thm:unaryshuffled}
	The problem \intreg$(${\sc Unary-Shuffled$\equiv_{\text{String}^\epsilon}$}$)$ is decidable.
\end{theorem}
\begin{proof}
	The language {\sc Unary-Shuffled$\equiv_{\text{String}^\epsilon}$} is context-free, since for a given word we only have to count the number of letters unequal to $\epsilon$ at the even and at the odd positions in the word. If those numbers are equal, the word is in {\sc Unary-Shuffled$\equiv_{\text{String}^\epsilon}$}. This property can be checked by a deterministic pushdown automaton and hence the language is context-free. Since the context-free languages are closed under intersection with regular languages and have a decidable emptiness problem \cite{hopcroft1969formal}, the problem \intreg$(${\sc Unary-Shuffled$\equiv_{\text{String}^\epsilon}$}$)$ is decidable, too.
\end{proof}

The \intreg\ problem becomes also decidable if we get rid of the shuffled encoding. The following problem (over an arbitrary large alphabet) has a decidable \intreg\ problem as well.

\begin{definition}[{\sc Sequential$\equiv_{\text{String}^\epsilon}$}]
	\ \\
	\textit{Given:} A word $w = s_{1}s_{2}\dots s_{n} \$ t_1t_2 \dots t_{n'}$ such that $s_i, t_i \in \Sigma \cup \{\epsilon \}$.\\
	\textit{Question:} Is $h(s_1s_2\dots s_n) = h(t_1t_2\dots t_{n'})$ where $h: \left(\Sigma \cup \{\epsilon, \$ \}\right)^* \to \left(\Sigma \cup \{\$\}\right)^*$ is an erasing homomorphism which leaves all symbols in $\Sigma \cup \{\$\}$ unchanged and deletes the $\epsilon$-symbols.
\end{definition}
\begin{theorem}\label{thrm:seqstring}
	The problem \intreg$(${\sc Sequential$\equiv_{\text{String}^\epsilon}$}$)$ is decidable.
\end{theorem}
\begin{proof}
	We define for every pair of states of the automaton $A=(Q, \Sigma, \delta, q_0, F)$ the set of sub-words, which can be read before the $\$$-symbol, between which the $\$$-symbol can be read, and which can be read after the $\$$-symbol.
	\begin{align*}
	&S_{q_0, q} := \{w \in  \left(\Sigma\cup\{\epsilon\}\right)^* \mid \delta(q_0, w) = q \wedge \exists q' \in Q : \delta(q, \$) = q'\}\\
	&\$_{q,q'} := \{\$ \mid \delta(q, \$) = q' \}\\
	&T_{q, q_f} := \{w \in \left(\Sigma\cup\{\epsilon\}\right)^* \mid \delta(q, w) = q_f \wedge q_f \in F \wedge \exists q' \in Q : \delta(q', \$) = q \}
	\end{align*}
	
	Therefore, $R$ can be written as $R = \bigcup_{q, q' \in Q,\,q_f\in F} S_{q_0, q}\$_{q, q'}T_{q', q_f}$.
	Since $A$ is a DFA, there are only finitely many sets $S_{q_0, q}$, $\$_{q,q'}$, and $T_{q, q_f}$ and all of them are regular, since we easily can alter the automaton $A$ to obtain finite automata for each of those languages.
	
	Let $h \colon \left(\Sigma \cup \{\epsilon, \$\}\right)^* \to {\left(\Sigma \cup \{\$\}\right)}^*$ be a homomorphism mapping every symbol form $\Sigma \cup \{\$\}$ to itself and deleting the $\epsilon$-symbols. 
	For every pair of states, the languages $h\left(S_{q_0, q}\right)$, $h\left(\$_{q,q'}\right)$, and $h\left(T_{q, q_f}\right)$ are regular. 
	For two regular languages the intersection emptiness problem is decidable, i.e.\ it is decidable, whether both languages contain a common word \cite{hopcroft1969formal}.
	
	It holds that {\sc Sequential$\equiv_{\text{String}^\epsilon}$}$\cap R \neq \emptyset$ if and only if there exists states $q,q',q_f \in Q$ and a word $v \in \Sigma^*$ such that $v \in h\left(S_{q_0, q}\right) \wedge \$ \in h\left(\$_{q,q'}\right) \wedge v \in h\left(T_{q', q_f}\right)$. Since there are only finitely many states in $Q$, and the membership and intersection-emptiness problems for regular languages are decidable, we can simply test the right-hand condition for every combination of states to decide whether there exists a word $v \in \Sigma^*$ fulfilling the condition.
\end{proof}

Since \intreg$(${\sc Sequential$\equiv_{\text{String}^\epsilon}$}$)$ is already decidable for an arbitrary alphabet it is also decidable for a unary alphabet, hence the problem \intreg$(${\sc Unary-Sequential$\equiv_{\text{String}^\epsilon}$}$)$ is decidable, too.
\begin{remark}
	We have shown that the problem of {\sc Shuffled$\equiv_{\text{RegEx}}$} has an undecidable $\intreg$-problem. The decidability statuses of the $\intreg$-problems for the problems of {\sc Sequential$\equiv_{\text{RegEx}}$}, {\sc Unary-Shuffled$\equiv_{\text{RegEx}}$}, and {\sc Unary-Sequential$\equiv_{\text{RegEx}}$} which are defined similarly to the variations of the $\equiv_{\text{String}^\epsilon}$-problems are still open.
	
	For the problem of {\sc Sequential$\equiv_{\text{RegEx}}$} the reduction from {\sc PCP} fails because we cannot describe the set of all eventual solutions of the {\sc PCP}-instance by a \emph{regular set} of regular expressions. The corresponding list-items in one possible solution are arbitrarily far apart from each other, because the two regular expressions are not in a shuffled but sequential encoding. Therefore, the relation between the corresponding list elements can no longer be generated by a \emph{regular expression}, even if we allow padding $\epsilon$-symbols.
	For the problems {\sc Unary-Shuffled$\equiv_{\text{RegEx}}$} and {\sc Unary-Sequential$\equiv_{\text{RegEx}}$} the reduction from {\sc PCP} fails to prove undecidability of the $\intreg$-problem because the {\sc PCP}-problem over unary alphabets is decidable \cite{rudnicki2003post}.
	
	At the attempt to prove decidability of the $\intreg$-problem for the above listed problems, we have to deal with the fact that we cannot restrict the given regular language to a regular language which only contains correctly encoded problem instances as we can do in all cases of proven $\intreg$ decidability \cite{guler2018deciding,DBLP:conf/dcfs/000219,Wolf:Thesis:2018}. 
	Allowing operators like alternation and star in the regular expressions automatically brings brackets to the regular expressions and therefore the language of all correctly encoded regular expressions is no longer a regular set. The fact that we cannot restrict the given regular language $R$ to a regular language of a known shape, containing the interesting subset of $R$, makes using pumping arguments difficult.
\end{remark}
\section{An Undecidable Intreg Problem About Graphs}
In this section we consider a graph problem with an undecidable \intreg-problem which contrasts the decidability results for graph problems in~\cite{DBLP:journals/corr/abs-2003-05826}.

For a word $w\in \Sigma^*$ and a letter $\sigma\in \Sigma$ we denote with $w_{|\sigma}$ the number of $\sigma$'s in~$w$. We consider directed hyper-multi-graphs with self loops and $2$ to $4$ vertices per edge. 
More formally, we consider graphs of the form $G=(V,E)$, where $V$ is a set of vertices and $E$ a set of edges together with the function $f_E \colon E \to V^{[2..4]}$ which assigns each edge with a tuple consisting of 2 to 4 vertices incident to this edge. An edge is called a \emph{loop} if all of its incident vertices are identical.
We encode $G$ by listing its edges separated by $\$$-signs. Vertices appearing in edges are encoded by strings of $\ta$'s separated by $\#$'s. To extract the encoded graph, we define the following decoding function.
For $m \in \mathbb{N}$, $1 \leq k_i \leq 3, i \leq m$ let
$$\decdhm\left(
\prod^m_{i = 1}\left(
 \ta^{p_i}
\prod^{k_i}_{j = 1}\left(
\#  \ta^{q_{i,j}}
\right)
\$
\right)\right) = G\,,$$
where $G = (V, E)$ with 
$V=\{v_{p_i}\mid i\in [m]\}\cup \bigcup\{v_{q_{i,j}}\mid j \leq k_i\},\ 
E=\{e_1, e_2, \dots, e_m\}$, $f_E \colon e_i \mapsto (v_{p_i},v_{q_{i,1}}, ..., v_{q_{i,k_i}})$.
We present a graph-problem over this class of graphs for which its \intreg-problem is undecidable by encoding the sets of derivation-trees of two given context-free grammars in Chomsky-normal-form (CNF for short) into a regular set of directed hyper-multi-graphs. The languages of the two grammars will share a common word $w$ if and only if the intersection of the constructed regular language with the graph-problem is non-empty and contains the two derivations of $w$.
We call $\mathcal{G} = (V, T, P, S)$ a context free grammar in CNF, if $V$ is a finite set of variables, $T$ a finite set of terminal, $P$ a set of derivation rules of the from $A \to BC$ or $A \to a$ with $A,B,C \in V$, $a\in T$, and $S\in V$ the start variable. 
We first give the construction and then define the graph problem \textsc{Embedded Derivation Trees}, \textsc{\GraphProblemFancyName} for~short.
\begin{theorem}
	\label{thm:graph}
	\intreg(\textsc{\GraphProblemFancyName}) is undecidable.
\end{theorem}
\begin{proof}
	Let $\mathcal{G}_1 = (V_1, T, P_1, S')$, $\mathcal{G}_2 = (V_2, T, P_2, S'')$ be two context free grammars in CNF. We alter them, by introducing two new variables $S_1$ and $S_2$, to the grammars  $\mathcal{G}_1 = (V_1 \cup \{S_1\}, T, P_1 \cup \{S_1\to S'\}, S_1)$, $\mathcal{G}_2 = (V_2 \cup \{S_2\}, T, P_2 \cup \{S_2 \to S''\}, S_2)$. From now on we will identify $V_1 \cup \{S_1\}$ as $V_1$ and $V_2 \cup \{S_2\}$ as $V_2$. W.l.o.g.~assume that the sets $V_1, V_2$ are disjoint and both grammars share the terminal alphabet $T$.
	Note that the following problem is undecidable: Is there a word $w \in T^*$ such that $w$ can be derived from $\mathcal{G}_1$ and from $\mathcal{G}_2$?
	
	Let $m_1 = |V_1|$, $m_2 = |V_2|$, $t = \frac{|T| (|T|+1)}{2}$, 
	$n = m_1 + m_2 + 2t + 2$.
	We construct a regular set $R = R_1 \cdot R_2$, where $R_1$ (respectively $R_2$) is defined as follows:
	We fix an order on the elements in $V_1, V_2, T$ such that $S_1$ is the first element in $V_1$ and $S_2$ is the first element in $V_2$. We refer to the $i$'th element of a set $\mathcal{S}$ as $\mathcal{S}[i]$.
	Let $V_1[s'] = S'$, $V_2[s''] = S''$ for integers $s'$ and $s''$. 
	For the derivation rule $S_1 \to S'$ in $\mathcal{G}_1$ we define the regular expression $r_{s_1} =  \ta^0\#  \ta^1\# \ta^0\#  \ta^{s'}\$$ and for the derivation rule $S_2 \to S''$ in $\mathcal{G}_2$ we define 
	$r_{s_2} =  \ta^{n-1}\#  \ta^{m_1+1}\# \ta^{n-1}\#  \ta^{m_1+s''}\$$
	For every derivation rule $p_1 = V_1[i] \to V_1[j]V_1[k]$, $i, j, k \leq m_1$ in $P_1$ we define the regular expression 
	$r_{p_1} =  \ta^0 \# \ta^i(\ta^n)^*\#  \ta^j(\ta^n)^* \#  \ta^k(\ta^n)^*\$.$
	For a derivation rule $p_2 = V_2[i] \to V_2[j]V_2[k]$, $i, j, k \leq m_2$ in $P_2$ we 
	define 
	$r_{p_2} = 
	 \ta^{n-1} \# \ta^{m_1+i}(\ta^n)^*\#  \ta^{m_1+j}(\ta^n)^* \# \ta^{m_1+k}(\ta^n)^*\$$.
	%
	We define for each $j \in [|T|]$ and $b \in \{1, 2\}$ a regular expression which encodes a cycle of length $j$ consisting of binary edges. We call these cycles \emph{leave-cycle} later. 
	\begin{align*}
	r_{lc}^{j{b}}= 
	%
	&\prod_{k=m_1+m_2+({b}-1)t+\frac{j(j+1)}{2}}^{m_1+m_2+({b}-1)t+\frac{(j+1)(j+2)}{2}-2}
	\left(  
	\ta^{k}(\ta^n)^* \# \ta^{k+1}(\ta^n)^*\$\right)\\
	&
	\ta^{m_1+m_2+({b}-1)t+\frac{(j+1)(j+2)}{2}-1}(\ta^n)^* \# \ta^{m_1+m_2+(b-1)t+\frac{j(j+1)}{2}}(\ta^n)^*\$
	\end{align*}
	For derivation rules $q_1 = V_1[i] \to T[j]$ in $P_1$ and $q_2 = V_2[i] \to T[j]$ in $P_2$ we define: 
	$r_{q_1} = 
	\ta^0\#  \ta^i(\ta^n)^*\#  \ta^{m_1+m_2+\frac{j(j+1)}{2}}(\ta^n)^*\$ r_{{lc}}^{j1}$, and\\
	$r_{q_2} = 
	\ta^{n-1} \#  \ta^{m_1+i} (\ta^n)^* \#  \ta^{m_1+m_2+t+\frac{j(j+1)}{2}} (\ta^n)^* \$r_{lc}^{j2}$.
	We are now ready to define $R_1$ and $R_2$. We set 
	$R_i = r_{s_i}\left(\bigcup_{p_i \in P_i} r_{p_i} \bigcup_{q_i \in P_i}r_{q_i}\right)^*$ for $i \in \{1, 2\}$.
	
	We now define our graph property such that it filters out the encoded graphs in the regular set which consists of two derivation trees, one from $\mathcal{G}_1$ and one from $\mathcal{G}_2$, which derivate the same word. It is helpful to consider Figure~\ref{fig:derivation} while reading though the following arguments.
	
	\begin{definition}[Multi-Graph Embedding $\phi$]
		\ \\
		Let $G=(V, E)$ be a directed hyper-multi-graph such that each edge contains two, three, or four vertices.
		The multi-graph embedding $\phi$ maps $G$ onto a multi-graph $\phi(G) = G_m = (V_m, E_m)$ with $E_m \subseteq V_m \times V_m$ in the following way: $V_m = V$, for a 4-nary edge $(a, b, c, d) \in E$ we add the edges $(a, c), (b, c), (b, d)$ to $E_m$.
		For a ternary edge $(a, b, c) \in E$ we add the edges $(a, c), (b, c)$ to $E_m$. Binary edges are simply added to $E_m$.
	\end{definition}
	\begin{definition}[\textsc{\GraphProblemFancyName}]
		\ \\
		\textbf{Input:} A directed hyper-multi-graph $G = (V, E)$ with $f_E \colon E \to V^{[2..4]}$.\\
		\textbf{Question:} 
		Does the multi-graph embedding $\phi(G)$ consists fo two connected components $G_{1m}$ and $G_{2m}$ such that the following holds.
		$G_{1m}$ contains a vertex $v_1$ and $G_{2m}$ contains a vertex $v_2$, such that $G_{1m}\backslash\{v_1\}$ and $G_{2m}\backslash\{v_2\}$ are (directed) binary trees (where edges are pointing from parents to children) in which the leave layer consists of directed cycles (called \emph{leave-cycles}). Exactly the root-node and the parents of leave-cycles have out-degree one, all other nodes which are not part of a leave-cycle have out-degree 2. 
		The node $v_1$ is connected to exactly one child of each parent node in $G_{1m}\backslash\{v_1\}$ except for the root as here $v_1$ is pointing to the root and not to the child. The only node pointing towards $v_1$ is the root of $G_{1m}\backslash\{v_1\}$. The node $v_1$ has no further connections. The same holds for $v_2$ with respect to $G_{2m}\backslash\{v_2\}$. If $G_{1m}$ and $G_{2m}$ are drawn such that $v_1$ and $v_2$ always point to the left child, then the sequence of lengths of the leave-cycles of $G_{1m}$ and $G_{2m}$ (read from left to right) must coincide. Both graphs must not contain multi-edges and loops are only allowed as a leave-cycle. The leave-cycles are not connected to each other.
	\end{definition}
	We will first argue that for any $w \in R$ with $\decdhm(w)$ being a positive instance of \textsc{\GraphProblemFancyName} the sub-graphs $G_{1m}$ and $G_{2m}$ of $\phi(\decdhm(w))$ must correspond to two derivation trees, one for $\mathcal{G}_1$ and one for $\mathcal{G}_2$.
	
	Note that for $i \equiv 0 \pmod{n}$ and $j \equiv -1 \pmod{n}$ the only vertex labels $\ta^i(\#|\$)$ and $\ta^j(\#|\$)$ which can be a factor of a word in $R$ are $\ta^0\#$ and $\ta^{n-1}\#$. Especially, there are no factors of the form 
	$(\ta^n)^k(\#|\$)$ or $(\ta^{n-1})(\ta^n)^k(\#|\$)$ with $k > 0$ and the vertex $\ta^0\#$ is only appearing in sub-graphs corresponding to derivation rules of $\mathcal{G}_1$ whereas $\ta^{n-1}\#$ is only appearing in sub-graphs corresponding to derivation rules of $\mathcal{G}_2$. Hence, for a graph $G=\phi(\decdhm(w))$ with $w\in R$ in order to consist of two disjoint graph $G_{1m}$ and $G_{2m}$ one of them must contain a vertex encoded by $\ta^0\#$ and hence be constructed by $\mathcal{G}_1$ and the other one must contain a vertex encoded by $\ta^{n-1}\#$ and be constructed by $\mathcal{G}_2$ as one of this elements is part of any regular expression $r_{s_i}, r_{p_i}, r_{q_i}$.
	
	Note that by the definition of $R_1$ and $R_2$ each string $w \in R$ contains exactly one factor encoding the derivation rule $S_1 \to S'$ and one factor encoding the derivation rule $S_2 \to S''$. There are no other derivation rules in which $S_1$ or $S_2$ appear. Hence, there will be no factor $ \ta^1(\ta^n)^k\#$ and $ \ta^{m_1+1}(\ta^n)^k\#$ in $w$ with $k > 0$. 
	%
	As $r_{s_1}$ and $r_{s_2}$ are the only regular expressions allowing to create an edge, the multi-edge embedding of which creates an arc pointing towards $v_1$ (encoded by $ \ta^0 \#$) and $v_2$ (encoded by $\ta^{n-1}\#$), the root of the tree\footnote{We interpret the cycles in the leave-layer as leaves and hence interpret the graph as a tree despite the fact that it contains cycles.} $G_{m1}\backslash\{v_1\}$  necessarily be $S_1$ with the child $S'$ and the root of the tree $G_{m2}\backslash\{v_2\}$ needs to be $S_2$ with child $S''$. The sub-tree starting in $S'$ will then correspond to a derivation tree of $\mathcal{G}_1$, when the leave-cycle of length $i$ is interpreted as the $i$-th letter of $T$, and the sub-tree starting in $S''$ will correspond to a derivation tree of $\mathcal{G}_2$. 

	W.l.o.g.\ we will focus on $G_{m1}$.
	By the definition of \textsc{\GraphProblemFancyName} $G_{1m}$ must be a binary tree with leave-cycles and by previous arguments have the root $S_1$ with the single child $S'$, all other internal parent-nodes have out-degree two. As $G_{1m}$ does not contain any multi-edges and loops are only allowed as leave-cycles, every inner node $V_{1m}[i]$ has exactly two appearances of its encoding $\ta^k\#$ in $w$ namely in the 4-nary hyper-edge where it appears as one child and in the 4-nary hyper-edge where it appears as the parent node. The first edge corresponds to a derivation where the variable represented by $V_{1m}[i]$ is 'created' and the second edge to a derivation where this variable is 'consumed'. Hence, following the inner nodes gives us a derivation tree for derivations of the form $A \to BC$ in $P_1$. 
	
	Every member of a cycle of length $k$ has its own domain of representatives which is disjoint with the domains of the other elements of the cycle. It is also disjoint with the domains of elements of cycles with different lengths. 
	We show that if $w \in R$  encodes a positive instance of \textsc{\GraphProblemFancyName}, then for each leave-cycle $C_i$ of length $k$ which is the child of a node $V_{1m}[j']$ encoding the variable $V_1[j]$, the description of $C_i$ is created by the regular expression $r_{q_1}$ which encodes the derivation $V_1[j] \to T[k]$. 
	Let $C_i[1]$ be the first node in the cycle $C_i$, i.e., the child of $V_{1m}[j']$. Then, there is a factor $uv$ in $w$ encoding the edge $(V_{1m}[j'], C_i[1])$ where $u$ encodes the node $V_{1m}[j']$ and $v$ encodes the node $C_i[1]$. We know that $u_{|\ta} \equiv j \pmod{n}$ and $v_{|\ta} \equiv m_1+m_2+\frac{k(k+1)}{2} \pmod{n}$. By the definition of $R$ the node $C_i[1]$ can only point to a node $C_i[2]$ encoded by a factor $x$ for which $x_{|\ta} \equiv m_1+m_2+\frac{k(k+1)}{2}+1 \pmod{n}$ for $k \geq 2$, or to itself for $k = 1$, but the only regular expressions which allow to create such a factor correspond to derivation rules $A \to T[k]$.  Indeed single edges between the same type of nodes (where the number of $\ta$'s have the same remainder modulo $n$) can be exchanged between different derivation rules with the same letter $T[k]$ on the right side 
	 but the number of nodes in a cycle can not be altered that way without losing a cycle structure or introducing forbidden multi-edges.
	Hence, we can assume that $C_i$ is created by a regular expression encoding the derivation rule $V_1[j] \to T[k]$.
	
	Replacing cycles of length $k$ by the corresponding letters $T[k]$ completes our derivation tree constructed by the inner nodes with derivations of terminals. As the trees $G_{1m}$ and $G_{2m}$ have the same sequence of cycle lengths in the leave-level if the child pointing to $v_1$, respectively $v_2$ is drawn as the left child, the constructed derivations derive the same word.
	
	For the other direction we can construct a word $w\in R$ for which the graph $\decdhm(w)$ is a positive instance of \textsc{\GraphProblemFancyName} from the two derivations of a common word $x \in \lang(\mathcal{G}_1) \cap \lang(\mathcal{G}_2)$ in the following way. We begin with the derivation tree for $\mathcal{G}_1$ and continue in the same way for the derivation tree for $\mathcal{G}_2$. We go through the derivation tree from the root to the leaves and in every level from left to right. It is clear how to encode the first derivation step $S_1 \Rightarrow S'$. For all other derivation steps we define an index $k$ and increase it with every new considered appearance of a variable in a derivation such that $k = 0$ for $S'$ as the child of the root. Then, we encode a node with variable $V_1[i]$ by a factor $ \ta^{i+kn}\#$. With a consistent variable encoding, the regular expressions $r_{p_1}$ tells us how to encode each derivation.
	For the terminal derivations we go from left to right. Let $k'$ be the index (in this sequence) of a terminal $T[j]$ derived from a variable $V_1[i]$ beginning with $k= 1$. Then, we use the regular expression $r_{q_1}$ corresponding to the derivation $q_1 = V[i] \to T[j]$ to encode a cycle where each element of the cycle is encoded by a factor of the form $\ta^{\ell + k'n}\#$ where $\frac{j(j+1)}{2} \leq \ell < \frac{(j+1)(j+2)}{2} < n$.
	Clearly $w$ encodes a positive instance of \textsc{\GraphProblemFancyName}.
\end{proof}

\begin{figure}[H]
	\centering
	\begin{tikzpicture}
	\node (S1) at (0,0) {$S_1$};
	\node[below of=S1] (S') {$S'$};
	\node[below left of=S'] (A1) {$A$};
	\node[below right of=S'] (B1) {$B$};
	\node[below left of=A1] (C1) {$C$};
	\node[right of=C1] (D1) {$D$};
	\node[below right of=B1] (C2) {$C$};
	\node[left of=C2] (B2) {$B$};
	\node[below of=C1] (11) {$1$};
	\node[below of=D1] (21) {$2$};
	\node[below of=B2] (31) {$3$};
	\node[below of=C2] (12) {$1$};
	
	\path
	(S1) edge (S')
	(S') edge (A1)
	(S') edge (B1)
	(A1) edge (C1)
	(A1) edge (D1)
	(B1) edge (B2)
	(B1) edge (C2)
	(C1) edge (11)
	(D1) edge (21)
	(B2) edge (31)
	(C2) edge (12);
	
	\node (S2) at (6,0) {$S_2$};
	\node[below of=S2] (S'') {$S''$};
	\node[below left of=S''] (E1) {$E$};
	\node[below right of=S''] (F1) {$F$};
	\node[below left of=F1] (G1) {$G$};
	\node[below right of=F1] (H1) {$H$};
	\node[below left of=H1] (H2) {$H$};
	\node[below right of=H1] (E2) {$E$};
	\node[below of=E1] (13) {$1$};
	\node[below of=G1] (23) {$2$};
	\node[below of=H2] (33) {$3$};
	\node[below of=E2] (14) {$1$};
	
	\path
	(S2) edge (S'')
	(S'') edge (E1)
	(S'') edge (F1)
	(F1) edge (G1)
	(F1) edge (H1)
	(H1) edge (H2)
	(H1) edge (E2)
	(E1) edge (13)
	(G1) edge (23)
	(H2) edge (33)
	(E2) edge (14);

	\end{tikzpicture}
	\ \\
	
	\vspace{2cm}
	\scalebox{0.5}{
		\begin{tikzpicture}
		
		\tikzset{every node/.style={state, node distance=3cm, minimum size=1.6cm}}
		\node (0) at (-5, 0) {$0$};
		\node (S1) at (0,0) {${1}$};
		\node[below of=S1] (S') {${2}$};
		\node[below left of=S'] (A1) {${3+1 n}$};
		\node[below right of=S'] (B1) {${4+2 n}$};
		\node[below left of=A1] (C1) {${5+3 n}$};
		\node[right of=C1] (D1) {${6+4 n}$};
		\node[below right of=B1] (C2) {${5+5 n}$};
		\node[left of=C2] (B2) {${4+6 n}$};
		\node[below of=C1] (11) {${13+1 n}$};
		
		\node[below of=D1] (21) {${14+2 n}$};
		\node[below of=21] (212) {${15+2 n}$};
		
		\node[below of=B2] (31) {${16+3 n}$};
		\node[below of=31] (312) {${17+3 n}$};
		\node[right of=312] (313) {${18+3 n}$};
		
		\node[below of=C2] (12) {${13+4 n}$};
		
		\path[->,ultra thick]
		(S1) edge (S')
		(S') edge (A1)
		(S') edge (B1)
		(A1) edge (C1)
		(A1) edge (D1)
		(B1) edge (B2)
		(B1) edge (C2)
		(C1) edge (11)
		(11) edge[loop below] (11)
		(D1) edge (21)
		(21) edge[bend left] (212)
		(212) edge[bend left] (21)
		(B2) edge (31)
		(31) edge (312)
		(312) edge (313)
		(313) edge (31)
		(C2) edge (12)
		(12) edge[loop below] (12);
		
		\path[->,ultra thick, color=red]
		(0) edge[bend left] (S1)
		(S1) edge (0)
		(0) edge[bend left] (A1)
		(0) edge[bend left] (C1)
		(0) edge[bend left] (B2)
		(0) edge[bend right] (11)
		(0) edge[bend right] (21)
		(0) edge[bend right] (31)
		(0) edge[bend right] (12);
		
		\node (19) at (17, 0) {${25}$};
		\node (S2) at (12,0) {${7}$};
		\node[below of=S2] (S'') {${8}$};
		\node[below left of=S''] (E1) {${9+1 n}$};
		\node[below right of=S''] (F1) {${10+2 n}$};
		\node[below left of=F1] (G1) {${11+3 n}$};
		\node[below right of=F1] (H1) {${12+4 n}$};
		\node[below left of=H1] (H2) {${12+5 n}$};
		\node[below right of=H1] (E2) {${9+6 n}$};
		\node[below of=E1] (13) {${19+1 n}$};
		\node[below of=G1] (23) {${20+2 n}$};
		\node[below of=23] (232) {${21+2 n}$};
		
		\node[below of=H2] (33) {${22+3 n}$};
		\node[below of=33] (332) {${23+3 n}$};
		\node[right of=332] (333) {${24+3 n}$};
		
		\node[below of=E2] (14) {${19+4 n}$};
		
		\path[->,ultra thick]
		(S2) edge (S'')
		(S'') edge (E1)
		(S'') edge (F1)
		(F1) edge (G1)
		(F1) edge (H1)
		(H1) edge (H2)
		(H1) edge (E2)
		(E1) edge (13)
		(13) edge[loop below] (13)
		(G1) edge (23)
		(23) edge[bend left] (232)
		(232) edge[bend left] (23)
		(H2) edge (33)
		(33) edge (332)
		(332) edge (333)
		(333) edge (33)
		(E2) edge (14)
		(14) edge[loop below] (14);
		
		\path[->,ultra thick, color=red]
		(19) edge[bend right] (S2)
		(S2) edge (19)
		(19) edge[bend right] (E1)
		(19) edge[bend right] (G1)
		(19) edge[bend right] (H2)
		(19) edge[bend left] (13)
		(19) edge[bend left] (23)
		(19) edge[bend left] (33)
		(19) edge[bend left] (14);
		
		\end{tikzpicture}
	}
	\caption{In the top two derivation trees for the word 1231 are drawn. Below is the multi-graph embedding of a corresponding instance of \textsc{\GraphProblemFancyName}. The chosen orders of the variable and letters are: $S_1, S', A, B, C, D$; $S_2, S'', E, F, G, H$; $1, 2, 3$; $n = 26$. For every node the number of letters $\ta$ in its encoding is depicted.}
	\label{fig:derivation}
\end{figure}
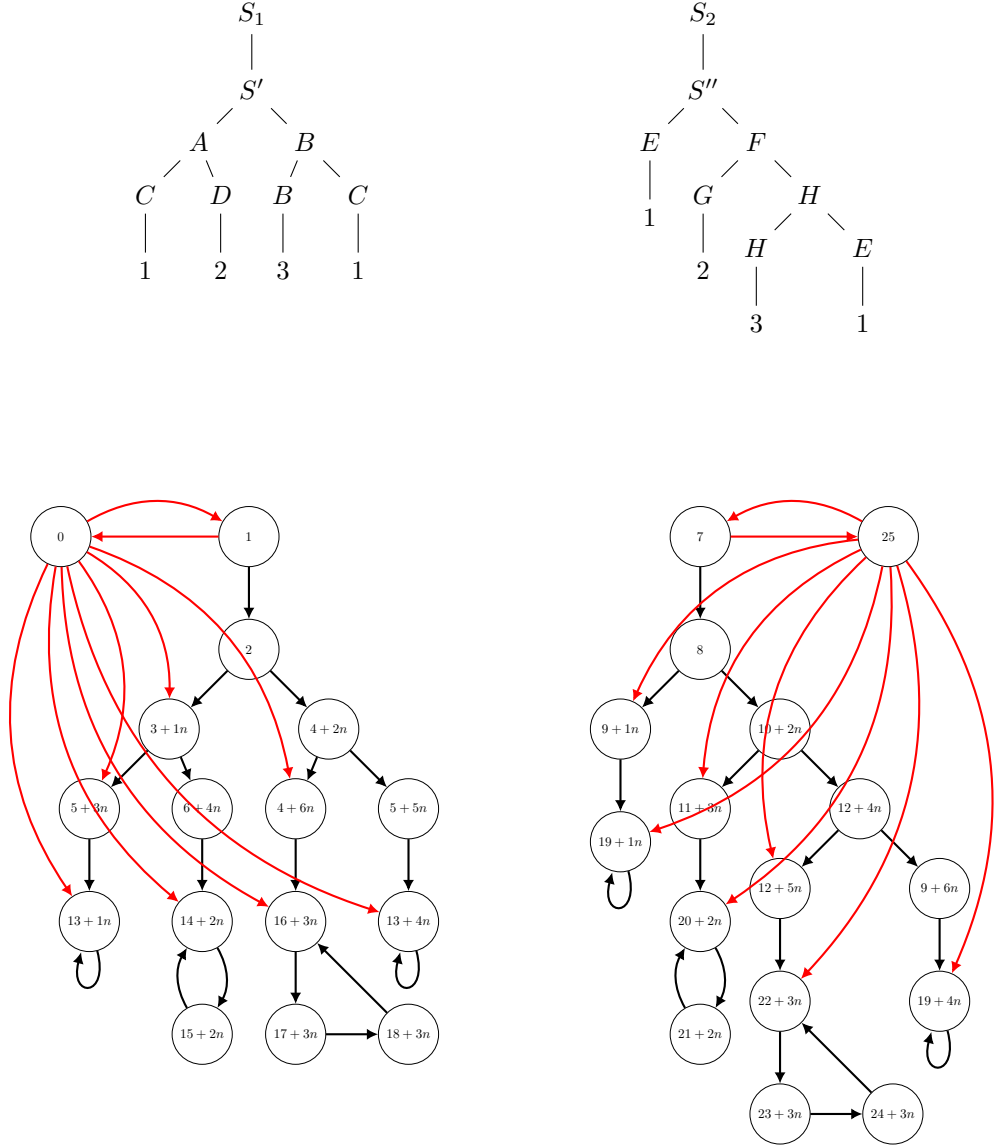
%
%
%
%
%
%
%
 \bibliographystyle{plain}
 \bibliography{mylit}

\end{document}